\documentclass[final,3p,times]{elsarticle}

\usepackage{amssymb}   
\usepackage{caption}   
\usepackage{subfigure} 
\usepackage{amsmath}   
\usepackage{setspace}  
\usepackage{algorithm} 
\usepackage{algpseudocode} 
\usepackage{makecell}  
\usepackage{multirow}  
\usepackage[justification=centering]{caption} 
\usepackage{multicol}
\usepackage{stfloats}  
\usepackage{color}     
\usepackage{amsthm}    
\usepackage{threeparttable}

\usepackage{dcolumn}

\usepackage{algorithm}
\usepackage{algpseudocode}

\newtheorem{definition}{Definition}
\newtheorem{example}{Example}
\newtheorem{theorem}{Theorem}

\begin{document}
\begin{frontmatter} 
%
\title{Approximate Order-Preserving Pattern Mining for Time Series}
%
%
%
	\author[1]{Yan Li}

	\author[2]{Jin Liu}
	
	\author[2]{Yingchun Guo}
	
	\author[2]{Jing Liu}
	
	\author[2]{{Youxi Wu}\corref{mycorrespondingauthor}}
	\ead{wuc567@163.com}
	
	\address[1]{School of Economics and Management, Hebei University of Technology, Tianjin 300401, China}
	
	\address[2]{School of Artificial Intelligence, Hebei University of Technology, Tianjin 300401, China}

\begin{abstract}
The order-preserving pattern mining can be regarded as discovering frequent trends in time series, since the same order-preserving pattern has the same relative order which can represent a trend. However, in the case where data noise is present, the relative orders of many meaningful patterns are usually similar rather than the same. To mine similar relative orders in time series, this paper addresses an approximate order-preserving pattern (AOP) mining method based on ($\delta$-$\gamma$) distance to effectively measure the similarity, and proposes an algorithm called AOP-Miner to mine AOPs according to global and local approximation parameters. AOP-Miner adopts a pattern fusion strategy to generate candidate patterns generation and employs the screening strategy to calculate the supports of candidate patterns. Experimental results validate that AOP-Miner outperforms other competitive methods and can find more similar trends in time series.

\end{abstract}

\begin{keyword}
sequential pattern mining \sep
time series \sep
order-preserving pattern \sep
($\delta$-$\gamma$) distance \sep
approximate pattern mining
\end{keyword}

\end{frontmatter}

%

\section{Introduction}
Sequential pattern mining \cite{ref1} is an important knowledge enhanced search method \cite{wutmis2} that has attracted great attention in recent years  \cite{ref3,tmisweak}. With the rapid development of network information technology, and in the era of big data, sequential pattern mining is becoming increasingly important and necessary. Due to its high efficiency and strong interpretability, it has been widely used in many applications, such as  virus similarity analysis \cite{ref6}, outlier sequence discovery \cite {leiduan, dse}, spatial co-location pattern discovery \cite {lizhenwang, lizhenwang2}, and  patient shock prediction \cite {ref7}.

As a common and important form of data, time series \cite{ref13, approximate} have been applied in numerous fields such as   multi-variate time series forecasting \cite{ref15,ref18, Jlin, zlibigdata},  stock prices prediction \cite{ref20}, and the analysis of EEG records \cite{ref21}. Unlike character sequences, a time series is a numerical sequence of data arranged in chronological order, which contains a large amount of regular information \cite {icde,Agrawaltkde}. To obtain valuable information, researchers have proposed many methods for the analysis of time series, such as three-way sequential pattern mining \cite{youxiins,minfan2020,youxitkdd},  neural network  method \cite {ref24,gnnAaai}, and transformer method \cite{ref25}.


However, for some applications, the trend of a time series may be more meaningful than its actual values. For example, in stock analysis, the trends shown by stocks are more worth studying than the actual prices. In temperature prediction, a change in temperature is more meaningful than the actual value of the temperature. As a new sequential pattern mining method, order-preserving pattern mining was proposed \cite{ref26}, which does not need to transform a time series into a character sequence to mine representative trends. This mining method can mine the most frequent subsequences with the same relative orders from a time series and ensure the continuity of the time series. The potentially changing laws of the time series can therefore be discovered, which can help users to better analyze and predict the time series data. 

The existing order-preserving pattern mining algorithm can only mine identical, short, order-preserving patterns \cite{ref26, oprminer}. Obviously, identical trends are only found in a situation where data noise is not allowed . More importantly, shorter patterns contain less information.   However, if we allow for some data noise, we can find longer order-preserving patterns with similar trends. For example, the order-preserving patterns of (4,1,3,2,5,6) and (5,2,3,1,4,6) are not the same, they have a high level of similarity with the order-preserving pattern $\textbf{p} = (4,2,3,1,5,6)$, since compared with \textbf{p}, the local error for each position satisfies $\le \delta$, and the global error satisfies $\le \gamma$. Hence, the order-preserving pattern \textbf{p} occurs twice in \textbf{s} and meets the minimum support threshold. The order-preserving pattern $\textbf{p} = (4,2,3,1,5,6)$ is therefore a frequent approximate order-preserving pattern (AOP).

Inspired by ($\delta$,$\gamma$)-approximate pattern matching \cite{ref44, ref28} and with the aim of overcoming the drawbacks of exact order-preserving pattern mining, this paper presents an AOP mining method based on the ($\delta$-$\gamma$) distance to effectively measure the similarity, and proposes the AOP-Miner algorithm for mining AOPs. The contributions of this paper are as follows.

\begin{enumerate} 
	\item   To mine some longer order-preserving patterns with similar trends, we focus on AOP mining in which the local error does not exceed $\delta$ and the overall error does not exceed $\gamma$.
	
	\item   We propose the AOP-Miner algorithm, which has two key steps: candidate pattern generation and pattern support calculation. AOP-Miner adopts a pattern fusion strategy to generate candidate patterns and employs screening and pruning strategies to calculate the support.
	
	\item   A large set of experimental results from real time series datasets is presented to verify that AOP-Miner yields better performance than alternative methods, and can find longer order-preserving patterns with similar trends.
\end{enumerate}

The rest of this paper is arranged as follows. Section \ref{section2} summarizes related work. Section \ref{section3} gives the relevant definition of AOP mining. Section \ref{section4} proposes the AOP-Miner algorithm and analyzes its time and space complexities. Section \ref{section5} verifies the performance of the AOP-Miner algorithm. Section \ref{section6} presents the conclusions of this paper.

\section{RELATED WORK} \label {section2}

A current research hotspot is the use of sequential pattern mining with the aim of quickly finding patterns that meet the specific needs of users. Various types of mining methods have been derived for different problems, such as  gap constraint mining \cite{ref29,apin2014, apin2022}, negative sequence mining \cite{ref32,tkdd2022}, rule mining \cite{icdmwpart}, high utility mining \cite{ref33, ref4, wuutil, utilityeswa}, and contrast pattern mining \cite{Mercericdm,ref39}. Sequential pattern mining methods can be categorized in many ways, such as  different types of mining data.





Based on the different types of mining data used, sequential pattern mining can be divided into classical sequential and time series pattern mining. Classical sequential pattern mining \cite{truongins} is mainly used for discrete sequences such as transaction datasets, webpage click-streams, DNA sequences, and gene sequences. A time series is a numerical sequence composed of continuously changing values. Due to the high dimensionality and continuity of time series data, it is very difficult to mine time series directly. It is therefore necessary to discretize the original numerical information into data in other domains through a series of transformations before performing mining. For example, the SAX algorithm \cite{ref40} was proposed to symbolize the time series, allowing various classical methods to be used to analyze the time series. However, this transformation will introduce new noise, since the essence of the process is to re-represent the time series, thus making it different from the original time series.


To effectively analyze the trends in time series data,  many methods were proposed to measure the distances of two time series, such as Euclidean distance and dynamic time warping (DTW) \cite{keoghdtw, kdd1994}. Among them, DTW is a famous method, since it can effectively handle shrink, scaling, and noise injection. However,  the DTW method is difficult to handle and find the same relative order in time series. For example, the DTW of two time series (13,11,18,23) and (16,12,19,25) is not zero, which means that they are different. However, the two time series have the same relative order (2,1,3,4), since 13 is the second lowest value in (13,11,18,23), 11 is the lowest, 18 is the third lowest, and 23 is the largest. To overcome this shortage, a simple method named order preserving pattern matching method was proposed {ref41}, in which the rank of each element in the sequence is adopted to represent the time series. This method can effectively find the same relative order  in the time series which can be regarded as patterns.  Inspired by order-preserving pattern matching, order-preserving pattern mining was proposed with the aim of mining frequently occurring trends from time series data \cite{ref26,oprminer} to help users analyze and predict trends in time series which can be used in disease spread analysis, temperature change analysis, and user behavior analysis. The abovementioned researches are based on exact matching, which can only mine frequent patterns with the same trends.

However, in the presence of data noise, it is necessary to use approximate pattern matching. For example, Paw et al. \cite{ref43} proposed the use of a matching accuracy based on the Hamming distance to measure similarity. If two strings have the same relative order after deleting up to {k} elements in the same positions, the matching is successful; however, the matching results for this problem are biased as the local similarity between the subsequence and the pattern cannot be measured. Mendivelso et al. \cite{ref44} proposed a similarity measurement method based on the ($\delta$-$\gamma$) distance to address this problem, and used local-global constraints to improve the accuracy of matching \cite{jinquandg}. All these methods are based on approximate order-preserving pattern matching.

Since current methods do not allow for data noise and only mine frequent patterns with identical trends, there is a need for a method of discovering frequent patterns with similar trends. Inspired by order-preserving pattern mining \cite{ref26,oprminer} and approximate order-preserving pattern matching \cite{ref44}, we investigate AOP mining and propose AOP-Miner to effectively mine frequent AOPs. The differences between the OPP-Miner method \cite{ref26} and AOP-Miner are as follows. OPP-Miner focuses on mining exact order-preserving patterns, which do not allow for data noise, while the approach in this paper is devoted to mining AOPs in a sequence with data noise. More importantly, OPP-Miner employs the order-preserving pattern matching method to calculate the support for each pattern, which requires repeated scanning of the whole sequence in the dataset. In contrast, AOP-Miner can effectively use the results of subpatterns to calculate the supports of superpatterns, an approach that can avoid the need for repeated scanning of the whole sequence. 

\section{Definitions} \label {section3}

\begin{definition}  \label{def1} 
	A time series \textbf{s} composed of {n} values can be expressed as \textbf{s}  = ${s}_1{s}_2 \cdots {s}_{i} \cdots {s}_{n}$ (1 $\le$ {i} $\le$ {n}), where $ {s}_{i} $ is called an element.
\end{definition}  

\begin{definition}  \label{def2} 
	The rank of element  $ {p}_{i} $ in pattern  $\textbf{p} ={p}_1{p}_2 \cdots {p}_{i} \cdots {p}_{m}$ with length  ${m} $ (1 $\le {i} \le {m} $) is  rankp($ {p}_{i} $) = 1+ $x $, where  $x$ means that there are  $x $ elements in  $\textbf{p} $ smaller than $ {p}_{i} $.
	\end {definition}
	
	\begin{definition}  \label{def3} 
		A relative order of pattern \textbf{p} with length {m} is an order-preserving pattern, represented by $r(\textbf{p}) = rankp({p}_1)rankp({p}_2)\cdots rankp( {p}_{m} $). 
	\end{definition}

	\begin{example}  \label{example2}           
		Suppose we have a pattern \textbf{p} = (12,15,10,13). Since element 12 is the second smallest in pattern \textbf{p}, its rank rankp(12) = 2. Similarly, element 15 is the largest in pattern \textbf{p}. Hence, rankp(15) = 4. The order-preserving pattern of \textbf{p} is then r(\textbf{p}) = (2,4,1,3).
		\end {example}
		
		\begin{definition}  \label{def4} 
			Given two time series \textbf{p} and \textbf{q} with length {m}, their order-preserving patterns are rankp(${p}_1$) rankp(${p}_2$)$\cdots$ rankp(${p}_m$) and rankp(${q}_1$)rankp(${q}_2$)$\cdots$ rankp(${q}_m$), respectively. The $\delta$ distance for \textbf{p} and \textbf{q} is the maximum value of the relative order, which is $d_{\delta}$(\textbf{p},\textbf{q}) = $\max_{{i}=1}^{{m}}$$|$rankp(${p}_j)$-rankp(${q}_j) |$. The $\gamma$ distance is the sum of the relative order, which is $d_{\gamma}$(\textbf{p},\textbf{q}) =$\sum_{{i}=1}^{{m}}$$|$rankp$({p}_j)$-rankp(${q}_j) |$.
		\end{definition}        
		
		\begin{definition}  \label{def5} 
			Suppose we have an order-preserving pattern \textbf{p} with length {m} and a subsequence \textbf{I} = ${s}_{t}{s}_{{t}+1}\cdots{s}_{{t}+{m}-1}$in the time series \textbf{s} (1 $\le {t} \le {n}-{m}+1$). If the order-preserving pattern r(\textbf{I}) and pattern \textbf{p} satisfy the ($\delta$-$\gamma$) distance constraint, i.e. $d_{\delta}$(\textbf{p},\textbf{I}) $\le \delta$ and $d_{\gamma}$(\textbf{p},\textbf{I}) $\le \gamma$, then \textbf{I} is a ($\delta$-$\gamma$) occurrence of \textbf{p} in \textbf{s}, where $\delta$ and $\gamma$ are two given, non-negative integers. To represent the occurrence concisely, in this paper, we use the first position of \textbf{I}  to represent the occurrence.
			\end {definition}
			
			\begin{definition}  \label{def6} 
				The number of ($\delta$-$\gamma$) occurrences of \textbf{p} in \textbf{s} is called the support, and is denoted by sup(\textbf{p},\textbf{s}). If the number of ($\delta$-$\gamma$) occurrences of \textbf{p} in \textbf{s} is not less than the {minimum} support threshold ({minsup}) (i.e., sup(\textbf{p},\textbf{s}) $\ge$ {minsup}), then \textbf{p} is a frequent ($\delta$-$\gamma$) order-preserving pattern.
				\end {definition}
				
				\begin{example}  \label{example3}
					Suppose we have $\delta = 1, \gamma = 2$, \textbf{p} = (2,1,4,3), and \textbf{s} = (12,15,10,13,11,18,23,9,26,20,13,16,12,19,25,20). All ($\delta$-$\gamma$) occurrences of \textbf{p} in \textbf{s} are as follows. For a time sub-series $\textbf{t}_\textbf{1}$  = (13,11,18,23) in \textbf{s}, the order-preserving r($\textbf{t}_\textbf{1}$) is (2,1,3,4). r($\textbf{t}_\textbf{1}$) and \textbf{p} satisfy ($\delta$-$\gamma$) order-preserving matching, since $|$2-2$|$ = 0 $\le \delta$, $|$1-1$|$ = 0 $\le \delta$, $|$4-3$|$ = 1 $\le \delta$, $|$3-4$|$ = 1 $\le \delta$, and $0+0+1+1 = 2 \le \gamma$. Therefore, \textbf{$\textbf{t}_\textbf{1}$} is a ($\delta$-$\gamma$) occurrence of \textbf{p} in \textbf{s}. Similarly, there are three other ($\delta$-$\gamma$) occurrences of \textbf{p} in \textbf{s}, which are $\textbf{t}_\textbf{2}$ = (23,9,26,20), $\textbf{t}_\textbf{3}$ = (16,12,19,25), and $\textbf{t}_\textbf{4}$ = (12,19,25,20).
					\end {example}

					\begin{definition}  \label{def7} 
						Our problem (AOP mining) is to find all frequent ($\delta$-$\gamma$) order-preserving patterns (AOPs).
						\end {definition}
						
						\begin{example}  \label{example4}  
							Suppose we have a time series \textbf{s} = (12,15,10,13,11,18,23,9,26,20,13,16,12,19,25,20), {$\delta$} = 1, {$\gamma$} = 2, and {minsup} = 4. Our aim is to mine all frequent AOPs. From Example 3, we know that (2,1,4,3) is a frequent AOP. Similarly, we know that all frequent AOPs are (1,2),(2,1),(1,2,3),(1,3,2),(2,1,3),(2,3,1),(3,1,2),(3,2,1),(1,2,3,4), (2,1,4,3), and (2,3,1,4).
							\end {example}
							
							It is worth noting that if $\delta$ = $\gamma$ = 0, then AOP mining can be seen as classical order-preserving pattern mining. Hence, AOP mining is a more general problem.

\section{AOP-Miner} \label {section4}

In AOP mining, there are two main steps: candidate pattern generation and calculation of the candidate pattern support. In Section \ref{subsect4.1}, we introduce the candidate pattern generation strategy. Section \ref{subsect4.2} proposes the candidate pattern support calculation strategy, and Section \ref{subsect4.3} presents AOP-Miner.

\subsection{Candidate pattern generation} \label{subsect4.1}

\subsubsection{Enumeration strategy}  

An enumeration strategy is a strategy that lists all possible situations one by one. Suppose pattern \textbf{p} is a frequent AOP with length {m} (1 $< {m}$). With an enumeration strategy, we will add a value at the end of \textbf{p}. There are therefore {m}+1 possible  relative order. Hence, a pattern can generate {m}+1 candidate patterns. An illustrative example is given below.

\begin{example}  \label{example5}   	
	Suppose we have two frequent patterns with length three, $ Fre_3 =$ \{(1,3,2),(2,1,3)\}. There are four candidate patterns based on the pattern (1,3,2). For example, if the newly added element is smaller than the original three elements, then the candidate pattern (2,4,3,1) is generated. In a similar way, the patterns (1,4,3,2), (1,4,2,3), and (1,3,2,4) are generated. Since each pattern can generate four candidate patterns, two frequent patterns can generate eight candidate patterns. The results are shown in Table \ref{tab1}.
	
	\begin{table}[ht] 
		\setlength{\abovecaptionskip}{0cm} 
		\caption{Generating candidate patterns using an enumeration strategy}\label {tab1}
		\centering
		\begin{tabular}{cccc}
			\hline
			${Fre}_3$       & Candidate patterns with length four ${C}_4$        \\ \hline
			
			(1,3,2)        & 	(2,4,3,1), (1,4,3,2), (1,4,2,3), (1,3,2,4)          \\
			
			(2,1,3)    & (3,2,4,1), (3,1,4,2), (2,1,4,3), (2,1,3,4)     \\
			\hline          				
		\end{tabular}
	\end{table}
	
	\end {example}        	
	
	Obviously, the higher the number of candidate patterns, the slower the mining speed. In the next subsection, we will introduce our pattern fusion strategy, which can effectively reduce the number of candidate patterns.
	
	\subsubsection{Pattern fusion strategy}  
	The pattern fusion strategy was proposed in \cite{ref26}, which uses the two concepts of order-preserving prefix and suffix subpatterns described as follows.
	
	\begin{definition}  \label{def8} 
		Suppose we have a pattern $\textbf{p} = ({p}_1{p}_2\cdots {p}_{m})$. After removing the last element $p_m$, the remaining subpattern $\textbf{q} = ({p}_1{p}_2\cdots {p}_{{m}-1})$ is the prefix subpattern of pattern \textbf{p}. Its order-preserving prefix subpattern is the order-preserving pattern of the prefix subpattern, and is denoted as prefixorder(\textbf{p}). Similarly, after removing the first element ${p}_1$, the remaining subpattern $\textbf{h} = ({p}_2{p}_3\cdots {p}_{{m}})$ is the suffix subpattern of pattern \textbf{p}, and its order-preserving suffix subpattern is the order-preserving pattern of the suffix subpattern, denoted as suffixorder(\textbf{p}). 
		\end {definition}
		
		\begin{definition}  \label{def9} 
			\textbf{ Pattern fusion:} Suppose there are two patterns $\textbf{p} = ({p}_1{p}_2 \cdots  {p}_{m})$ and $ \textbf{q} = ({q}_1{q}_2\cdots {q}_{m})$ with length {m}. If suffixorder(\textbf{p}) = prefixorder(\textbf{q}), then we can generate one or two superpatterns with length {m}+1 by pattern fusion of \textbf{p} and \textbf{q}, i.e., \textbf{t} = \textbf{p} $\oplus$ \textbf{q} or \textbf{t},\textbf{k} = \textbf{p} $\oplus$ \textbf{q} There are three possible cases:
			
			\textbf{Case 1}: ${p}_1 > {q}_{m}$. In this case, a superpattern \textbf{t} is generated. First, we let ${t}_1 = {p}_1+1$. We then compare ${q}_{j}$ ($1 \le {j} \le {m} $) with ${p}_1$. If ${q}_{j}>{p}_1$, then ${t}_{{j}+1} =  {q}_{j}+1$. Otherwise, ${t}_{{j}+1} = {q}_{j}$.
			
			\textbf{Case 2}: ${p}_1 = {q}_{m}$. In this case, two superpatterns \textbf{t} and \textbf{k} are generated. For pattern $\textbf{t} = ({t}_1{t}_2 \cdots {t}_{m}{t}_{{m}+1})$, let ${t}_1 = {p}_1+1$. We then compare ${q}_{j}$ $(1 \le {j} \le {m})$ with ${p}_1$. If ${q}_{j} > {p}_1$, then ${t}_{{j}+1} = {q}_{j}+1$. Otherwise, ${t}_{{j}+1} = {q}_{j}$. For pattern $\textbf{k} = ({k}_1{k}_2 \cdots  {k}_{m}{k}_{{m}+1})$, we let ${k}_{{m}+1} = {q}_{m}+1$. We then compare ${p}_{i}$ $(1 \le {i} \le {m})$ with ${q}_{m}$. If ${p}_{i} > {q}_{m}$, then ${k}_{i} = {p}_{i}+1$. Otherwise, ${k}_{i} = {p}_{i}$.
			
			\textbf{Case 3}: ${p}_1 < {q}_{m}$. In this case, a superpattern \textbf{t} is generated. First, let ${t}_{{m}+1} = {q}_{m}+1$. We then compare ${p}_{i}$ $(1 \le {i} \le {m})$ with ${q}_{m}$. If ${p}_{i} > {q}_{m} $, then ${t}_{i} = {p}_{i}+1$. Otherwise,  ${t}_{i} = {p}_{i}$.
			\end {definition}

			\begin{example}  \label{example6}
				Suppose we have an order-preserving pattern $\textbf{p} = (1,3,2)$ which can join with patterns $\textbf{q} = (2,1,3)$ and $ \textbf{r} = (3,2,1) $, since suffix(\textbf{p}) $= (3,2)$, prefix(\textbf{q})$ = (2,1) $, prefix(\textbf{r}) $= (3,2)$, and their order-preserving patterns are  suffixorder(\textbf{p}) = prefixorder(\textbf{q}) = prefixorder(\textbf{r}) $= (2,1)$.
				\end {example}

				For patterns \textbf{p} and \textbf{r},  suffixorder(\textbf{p}) = prefixorder(\textbf{r}) and ${p}_1 = {r}_3 = 1$, which satisfies Case 2. Hence, \textbf{p} and \textbf{r} can generate two patterns \textbf{t} and \textbf{k} with length four. We know that $\textbf{t} = (2,4,3,1)$ and $\textbf{k} = (1,4,3,2)$, as shown in Fig. \ref{fig.3}.
				
				\begin{figure}[h]   
					\setlength{\abovecaptionskip}{0cm} 
					\centering
					\includegraphics[width=0.4\linewidth]{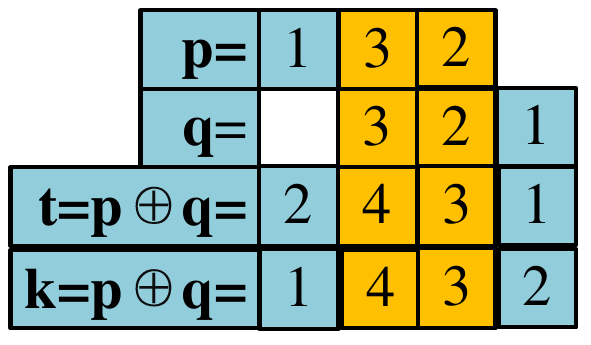}
					\caption{ An example of Case 2}
					\label{fig.3}
				\end{figure}

				We know that suffixorder(\textbf{p}) = prefixorder(\textbf{q}) $= (2,1)$ and $1 < 3$, which indicates Case 3. Hence, ${t}_4 = {q}_3+1 = 4$. We then compare ${p}_1, {p}_2$, and ${p}_3$ with ${q}_3$. Since ${p}_1, {p}_2$, and ${p}_3$ are less than ${q}_3$, $ ({t}_1{t}_2{t}_3) = ({p}_1{p}_2{p}_3) = (1,3,2) $ . Hence, $ {t} = (1,3,2,4) $, as shown in Fig. \ref{fig.2}.
				
				\begin{figure}[h]   
					\setlength{\abovecaptionskip}{0cm} 
					\centering
					\includegraphics[width=0.4\linewidth]{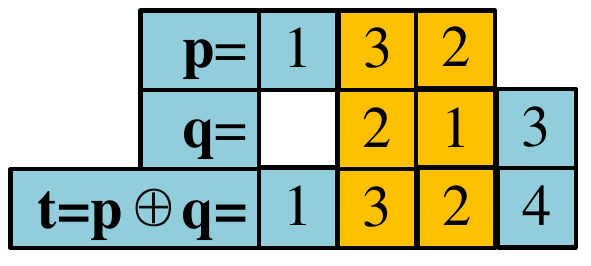}
					\caption{ An example of Case 3}
					\label{fig.2}
				\end{figure}

				\begin{example}  \label{example7}      
					We use the same frequent pattern set ${Fre}_3$ as in Example \ref{example5} to generate candidate patterns with length four using a pattern fusion strategy.
					
					From Example \ref{example6}, we know that (1,3,2) $\oplus$ (2,1,3) generates (1,3,2,4). Similarly, we know that (2,1,3) $\oplus$ (1,3,2) generates (3,1,4,2) and (2,1,4,3). The results are shown in Table \ref{tab2}.
					
					\begin{table}[ht] 
						\setlength{\abovecaptionskip}{0cm} 
						\caption{Generating candidate patterns using a pattern fusion strategy}\label{tab2}
						\centering
						\scalebox{0.97}{
							\begin{tabular}{cccc}
								\hline
								${Fre}_3$       & Candidate patterns with length four ${C}_4$        \\ \hline
								
								(1,3,2)       & 	(1,3,2,4)         \\
								
								(2,1,3)   & (3,1,4,2), (2,1,4,3)     \\
								\hline          				
						\end{tabular}}
					\end{table}
					\end {example}   
					
					By comparing Tables \ref{tab1} and \ref{tab2}, we can see that the enumeration strategy generates eight candidate patterns with length four, while the pattern fusion strategy generates only three candidate patterns. Examples \ref{example5} and \ref{example7} show that the pattern fusion strategy outperforms the enumeration strategy, since the pattern fusion strategy can effectively reduce the number of candidate patterns, thereby improving the mining efficiency of the algorithm.
					
					\subsection{Support calculation}\label{subsect4.2}
					\subsubsection{segtreeBA and bitBA} 
					
					To the best of our knowledge, only the segtreeBA and bitBA algorithms \cite{ref44} have been developed for calculating the support of a ($\delta$-$\gamma$) order-preserving pattern. To calculate the support of an order-preserving pattern, these two algorithms need to scan all possible candidate occurrences. Example \ref{example8} illustrates the principle of operation of these two algorithms.
					
					\begin{example}  \label{example8}
						Given a time series \textbf{s} = (12,15,10,13,11,18,23, 9,26,20,13,16,12,19,25,20), {minsup} = 4, $\delta$ = 1, $\gamma$ = 2, we can use the segtreeBA and bitBA algorithms to find the candidate occurrences of pattern (2,3,1,5,4).	
						
						The segtreeBA and bitBA algorithms scan the time series from the beginning to the end to enumerate all candidate occurrences one by one \cite{ref44}. For example, the first candidate occurrence is 1, and its corresponding subsequence is (12,15,10,13,11). There are 12 candidate occurrences, since the lengths of s and the pattern are 16 and 5, respectively, and 16-5+1 = 12. Thus, the last candidate occurrence is 12 and its corresponding subsequence is (16,12,19,25,20).
						\end {example}

						To calculate the support of a candidate pattern, both algorithms need to scan the database and compare each candidate occurrence with the pattern to determine whether it meets the approximate matching criterion. The database is therefore scanned multiple times, which greatly increases the calculation time. To solve this problem, we propose the Alar algorithm, which can quickly calculate the support of the candidate pattern without repeatedly scanning the database.
						
						\subsubsection{Alar algorithm} 
						
						In this section, we describe the Alar algorithm, which finds all frequent AOPs with length ${m}$+1 using the results for all frequent AOPs with length ${m}$. The Alar algorithm applies two key strategies to reduce the support calculation: a screening strategy to find the candidate occurrences of a superpattern based on the occurrences of subpatterns, and a pruning strategy to further reduce the number of candidate patterns. The principles of operation of these two strategies are as follows.
						
						In the pattern fusion strategy, each candidate pattern is generated by joining two frequent ($\delta$-$\gamma$) preserving patterns, i.e., \textbf{t} = \textbf{p} $\oplus$ \textbf{q}. Thus, we know that if $x$ is an occurrence of pattern \textbf{t}, then prefixorder($x$) must be an occurrence of pattern \textbf{p} and suffixorder($x$+1) must be an occurrence of pattern \textbf{q}. Hence, if $x$ is an occurrence of pattern \textbf{p} and $x$+1 is an occurrence of pattern \textbf{q}, then $x$ may be an occurrence of pattern \textbf{t}. Otherwise, $x$ is not an occurrence of pattern \textbf{t}. We can then develop the following screening strategy.
						
						\textbf{Screening strategy}: If $x$ is an occurrence of pattern \textbf{p}, but $x$+1 is not an occurrence of pattern \textbf{q}, then $x$ is not an occurrence of pattern \textbf{t}. Similarly, if $x$-1 is not an occurrence of pattern \textbf{p}, but $x$ is an occurrence of pattern \textbf{q}, then $x$-1 is not an occurrence of pattern \textbf{t}.

						\begin{theorem}\label{thm1}
							The screening strategy is correct.
						\end{theorem}
						\begin{proof}        
							\textbf{Proof}. The proof is by contradiction. We assume that $x$ is an occurrence of pattern \textbf{t}, and $x$ is an occurrence of pattern \textbf{p}, but that $x$+1 is not an occurrence of pattern \textbf{q}. Hence, $x$+1 is an occurrence of pattern \textbf{q}, since $x$ is an occurrence of pattern \textbf{t}. This contradicts the fact that $x$+1 is not an occurrence of pattern \textbf{q}. Similarly, we know that if $x$-1 is not an occurrence of pattern \textbf{p}, then $x$-1 is not an occurrence of pattern \textbf{t}.
						\end{proof}
						
						We use this screening strategy to effectively screen out feasible candidate occurrences. All occurrences of \textbf{p} and \textbf{q} are stored in arrays $A_p$ and $A_q$, respectively. Example \ref{example9} illustrates the principle of this strategy. 
						
						\begin{example}  \label{example9}
							We use the same data as in Example \ref{example4}, where (2,3,1,4) and (2,1,4,3) are two frequent patterns of length four. We know that (2,3,1,4) $\oplus$ (2,1,4,3) can generate the candidate pattern (2,3,1,5,4). 
							\end {example}
							
							For a length-four frequent pattern (2,3,1,4), there are four occurrences in \textbf{s}: \{1,3,6,11\}, which are stored in an array {$A_p$}. The corresponding subsequences are ($s_1s_2s_3s_4$), ($s_3s_4s_5s_6$), ($s_6s_7s_8s_9$), and ($s_{11}s_{12}s_{13}s_{14}$), and all of these subsequences are occurrences of pattern (2,3,1,4) with $\delta =1 $ and $\gamma = 2$. Similarly, there are four occurrences of pattern (2,1,4,3) in \textbf{s}: \{4,7,12,13\}, which are stored in {$A_q$}. Table \ref{tab3} shows the occurrences in Example \ref{example9}.
							
							\begin{table}[ht] 
								\setlength{\abovecaptionskip}{0cm} 
								\caption{Occurrences of frequent patterns in Example \ref{example9}}\label{tab3}
								\centering
								\scalebox{0.97}{
									\begin{tabular}{cccc}
										\hline
										Frequent pattern      & Occurrences with $\delta = 1 $ and $\gamma = 1 $        \\ \hline
										
										(2,3,1,4)       & \{1,3,6,11\}        \\
										
										(2,1,4,3)   & \{4,7,12,13\}    \\
										\hline          				
								\end{tabular}}
							\end{table}
							
							We now generate the candidate occurrences of pattern (2,3,1,5,4) based on the occurrences of (2,3,1,4) and (2,1,4,3). From Table \ref{tab3}, we know that 1 is an occurrence of (2,3,1,4), but 2 is not an occurrence of (2,1,4,3). Thus, 1 cannot be a candidate occurrence of (2,3,1,5,4) according to the screening strategy. Since 3 is an occurrence of (2,3,1,4) and 4 is an occurrence of (2,1,4,3), 3 can be a candidate occurrence of (2,3,1,5,4) according to the screening strategy. Similarly, we know that both 6 and 11 can be candidate occurrences. Hence, the corresponding subsequences of the candidate occurrences are (10,13,11,18,23), (18,23,9,26,20), and (13,16,12,19,25).
							
							We can see that after applying the screening strategy, there are only three candidate occurrences, far fewer than in Example \ref{example8}. The lower the number of candidate occurrences, the higher the efficiency of the algorithm.

							To further improve the efficiency, we also propose a pruning strategy to reduce the number of candidate patterns.
							
							\textbf{Pruning strategy}: If the number of candidate occurrences of a pattern is less than the threshold {minsup}, then the pattern can be pruned. 
							
							\begin{theorem}\label{thm2}
								The pruning strategy is correct.
								\end {theorem}
								\begin{proof}
									Obviously, the support of a pattern is not greater than the number of candidate occurrences, if the number of candidate occurrences of one pattern is less than {minsup}. Hence, the pattern cannot be an AOP and can be pruned.
								\end{proof}
								
								\begin{example}  \label{example10}
									We use the same data as in Example \ref{example9}. We know that pattern (2,3,1,5,4) has three candidate occurrences, which is less than ${minsup} = 4$. Hence, pattern (2,3,1,5,4) can be pruned.
									\end {example}
									
									The principle of the Alar algorithm is as follows. Alar judges whether any two frequent patterns \textbf{p} and \textbf{q}  with length {m} can be fused. If yes, then Alar generates the candidate patterns according to Definition \ref{def9}, and determines whether or not the candidate patterns are frequent using the Checking algorithm. The Alar algorithm is shown in Algorithm \ref{alg:1}.
									
									\begin{algorithm}[t] 
										\caption{Alar}
										\hspace*{0.02in} {\bf Input:} frequent patterns $ \textit{F}_m$, \textbf{s}, $ \delta $, $ \gamma $, \textit{m} \\
										\hspace*{0.02in} {\bf Output:} frequent patterns $\textit{F}_{m+1}$ and their corresponding occurrence arrays 
										\begin{algorithmic}[1]
											\For {each \textbf{p} in $\textit{F}_m$}
											\For {each \textbf{q} in $\textit{F}_m$}
											\If {suffix(\textbf{p}) $=$ prefix(\textbf{q})}
											\If {$\textit{p}_1 = \textit{q}_m $}
											\State \textbf{t} $\cup$ \textbf{k} $\leftarrow$\textbf{p} $\oplus$ \textbf{q}; 
											\State \textit{$F_{m+1}$} $\leftarrow$ \textit{$F_{m+1}$} $\cup$  Checking ($\textbf{t}$, $\textit{A}_p$, $\textit{A}_q,\textbf{s}, \delta, \gamma$,  $\textit{minsup}$) $\cup$ Checking ($\textbf{k}, \textit{A}_p, \textit{A}_q, \textbf{s}$, $\delta, \gamma, \textit{minsup}$); 
											\Else 
											\State \textbf{t} $\leftarrow$\textbf{p} $\oplus$ \textbf{q}; 
											\State \textit{$F_{m+1}$} $\leftarrow$ \textit{$F_{m+1}$} $\cup$  Checking($\textbf{t},\textit{A}_p, \textit{A}_q $, $\textbf{s}, \delta, \gamma$,  $\textit{minsup}$); 
											\EndIf
											\EndIf
											\EndFor
											\EndFor
											\State return $\textit{F}_{\textit{m}+1}$;
										\end{algorithmic}
										\label{alg:1}
									\end{algorithm}
									
									The Checking algorithm adopts the screening strategy to find candidate occurrences and applies the pruning strategy to decrease the number of candidate patterns. Finally, it uses the Matching algorithm to calculate the support of the candidate pattern. The Checking algorithm is shown in Algorithm \ref{alg:2}.
									
									\begin{algorithm}[t] 
										\caption{Checking}
										\hspace*{0.02in} {\bf Input:}
										superpattern \textbf{t}, occurrence arrays $\textit{A}_p$ and $\textit{A}_q$, \textbf{s}, $\delta$, $\gamma$ and \textit{minsup}  \\
										\hspace*{0.02in} {\bf Output:}
										frequent pattern \textbf{t} and its corresponding occurrence array $\textit{A}_t$
										\begin{algorithmic}[1]
											\For {for each occurrence \textit{occ} in $\textit{A}_p$}    
											\If { $\textit{occ}+1$ in $A_q$}
											\State $\textit{C}_t$.append(\textit{occ});
											\EndIf
											\EndFor
											\If {  len($\textit{C}_t) <\textit{minsup} $}
											\State  return Null; 
											\EndIf
											\State  $\textit{A}_t \leftarrow$ Matching($\textit{C}_t$, \textbf{t}, \textbf{s}, $\delta$, $\gamma$, \textit{minsup});
											\If { len($\textit{A}_t) <\textit{minsup} $ }
											\State return Null; 
											\EndIf 
											\State return $\textit{A}_t$;
											
										\end{algorithmic}
										\label{alg:2}
									\end{algorithm}
									
									The Matching algorithm uses the definitions to calculate the support and is shown in Algorithm \ref{alg:3}.
									
									\begin{algorithm}[t] 
										\caption{Matching}
										\hspace*{0.02in} {\bf Input:}
										$\textit{C}_t$, \textbf{t}, \textbf{s}, $\delta$, $\gamma$, \textit{minsup}   \\
										\hspace*{0.02in} {\bf Output:}
										occurrence array $\textit{A}_t$
										\begin{algorithmic}[1]
											\For {each occurrence \textit{occ} in $\textit{C}_t$ }   
											\State \textit{cand} $\leftarrow$ \textbf{s}[\textit{occ}: \textit{occ}$+len($\textbf{t})];
											\State   $\textbf{f}$ $\leftarrow$ r($\textit{cand}$);
											\If { $d_{\delta}(\textbf{f},\textbf{t}) \le \delta$ and $d_{\gamma}(\textbf{f},\textbf{t}) \le \gamma$ }
											\State    $\textit{A}_t$.append(\textit{occ});
											\EndIf
											\EndFor
											\State  return  $\textit{A}_t$;
										\end{algorithmic}
										\label{alg:3}
									\end{algorithm}
									
									\subsection{AOP-Miner algorithm}   \label{subsect4.3}  
									
									Pseudocode for AOP-Miner is given in Algorithm \ref{alg:4}.
									
									\begin{algorithm}[t] 
										\caption{AOP-Miner}
										\hspace*{0.02in} {\bf Input:}
										time sequence dataset \textbf{s}, local constraint $\delta$, global constraint $\gamma$, minimum support threshold \textit{minsup}  \\
										\hspace*{0.02in} {\bf Output:}
										frequent ($\delta$-$\gamma$) order-preserving pattern set $\textit{F}_m$
										\begin{algorithmic}[1]
											\State   Calculate the supports of candidate patterns (1,2) and (2,1) and store their occurrences in $A_{(1,2)}$ and $A_{(2,1)}$, respectively. If they are frequent, store them in $F_2$;
											\State    \textit{m}$\leftarrow$2;
											\While   { $\textit{F}_m$ != Null;}
											\State    $\textit{F}_{m+1}\leftarrow$Alar($\textit{F}_m,\textbf{s}, \delta, \gamma, \textit{misnup}$);
											\State      $\textit{m}\leftarrow\textit{m}+1$;
											\EndWhile
											\State return $\textit{F}_m$;
										\end{algorithmic}
										\label{alg:4}
									\end{algorithm}

									The AOP-Miner algorithm uses three key steps to mine all frequent AOPs.
									
									Step 1: The supports of candidate patterns (1,2) and (2,1) are calculated and frequent patterns are stored in ${{F}}_2$.
									
									Step 2: The Alar algorithm is applied to mine frequent AOPs with length $ {m}+1 $ using frequent AOPs with length ${m}$.
									
									Step 3: Step 2 is iterated until no new AOPs are found. 
									
									%
									%
									%
									%

\section{Experimental results and performance analysis}\label{section5}

In this section, we evaluate the performance of AOP-Miner. Section \ref{subsect5.1} explains the experimental environment and the datasets used, and Section \ref{subsect5.2} introduces the baseline methods. Section \ref{subsect5.3} verifies the performance of each strategy in AOP-Miner. 

\subsection{Dataset} \label{subsect5.1}

The experimental running environment was an Intel(R)Core(TM)i7-6560U processor, 2.20 GHZ CPU with 16 GB memory, Windows 10, a 64-bit operating system, and the program development environment was VS2017.

We used eight real-time series datasets to conduct our experiments, including oil, stock, air temperature,and air quality data. Oil datasets are the closing price of 1WTI crude oil and London Brent crude oil, which can be downloaded at http://www.fansmale.com/index.html. Stock datasets are the closing price of the Russell 2000 index and Nasdaq index, which can be downloaded at https://www.yahoo.com. The air temperature datasets are the average of the air temperature of Shunyi, Huairou, Changping and Tiantan, which can be downloaded at https://archive.ics.uci.edu/ml/datasets.php. The air quality datasets is the PM2.5 values of Beijing city which can be found at https://tianqi.2345.com. A description of each dataset is given in Table \ref{tab4}.

\begin{table}[ht] 
	\setlength{\abovecaptionskip}{0cm} 
	\caption{Occurrences of frequent patterns in Example \ref{example9}}\label{tab4}
	\centering
	\scriptsize
	\scalebox{0.97}{
		\begin{tabular}{cccc}
			\hline
			Dataset     & Name     & Type    & Length       \\ \hline
			
			SDB1      & 	1WTI (2000.07.26-2020.02.26)     & 	Oil   &  4,996  \\
			SDB2      & 	London Brent (2000.09.04-2020.02.26)    & 	Oil   &  5,000  \\
			SDB3      & 	Russell 2000 (1987.9.10-2019.12.27)     & 	Stock  &  8,141 \\
			SDB4      &    Nasdaq (1971.3.1-2019.10.31)  & 	Stock  &  12,279  \\
			SDB5      &    PRSA\_Data\_Shunyi (2015.1.1-2017.2.28)  & 	Air quality  &  18,960  \\
			SDB6      &    PRSA\_Data\_Huairu (2015.1.1-2017.2.28)   & 	Air quality  &  18,960  \\
			SDB7      & 	PRSA\_Data\_Changping (2013.3.1-2017.2.28)   & Air temperature  & 35,064 \\
			SDB8      & 	PRSA\_Data\_Tiantan (2013.3.1-2017.2.28)   & 	Air temperature  &  35,064  \\
			\hline          				
	\end{tabular}}
\end{table}

\subsection{Baseline methods} \label {subsect5.2}

To evaluate the performance of AOP-Miner, four alternative algorithms were proposed and a competitive algorithm was selected. The operation of each of these algorithms is described below.

1. seg-em-Miner and bit-em-Miner \cite{ref44}: To validate the efficiency of the candidate pattern generation and support calculation, we developed seg-em-Miner and bit-em-Miner, which used an enumeration strategy to generate the candidate patterns and the segtreeBA and bitBA algorithms to calculate the support, respectively.

2. em-Miner: To verify the performance of our pattern fusion strategy, we developed the em-Miner algorithm, which used the enumeration strategy to generate the candidate patterns. The support calculation method in em-Miner is the same as that of AOP-Miner .


3. Nopruning-Miner: To validate the performance of our pruning strategy, we developed Nopruning-Miner, which adopted a pattern joining strategy to generate the candidate patterns and applied the screening strategy to calculate the support.



\subsection{Performance of each strategy}\label{subsect5.3}

To verify the mining efficiency of each strategy in AOP-Miner, we assessed four competitive algorithms: seg-em-Miner, bit-em-Miner, em-Miner,  and Nopruning-Miner. Eight datasets were selected (SDB1 to SDB8). The experimental parameters were $\delta$ = 2 and $\gamma$ = 4. Since the lengths of the datasets were different, to mine reasonable number of patterns, {minsup} was set to 1000, 1000, 1600, 4000, 4000, 4000, 12000, and 12000 for SDB1–SDB8, respectively. These algorithms were able to find 15, 15, 18, 11, 12, 12, 13, and 11 AOPs from SDB1–SDB8, respectively. Comparisons of the number of candidate patterns and running time are shown in Figs. \ref{fig.4} and \ref{fig.5}.

\begin{figure}   
	\setlength{\abovecaptionskip}{0cm} 
	\centering
	\includegraphics[width=1.0\linewidth]{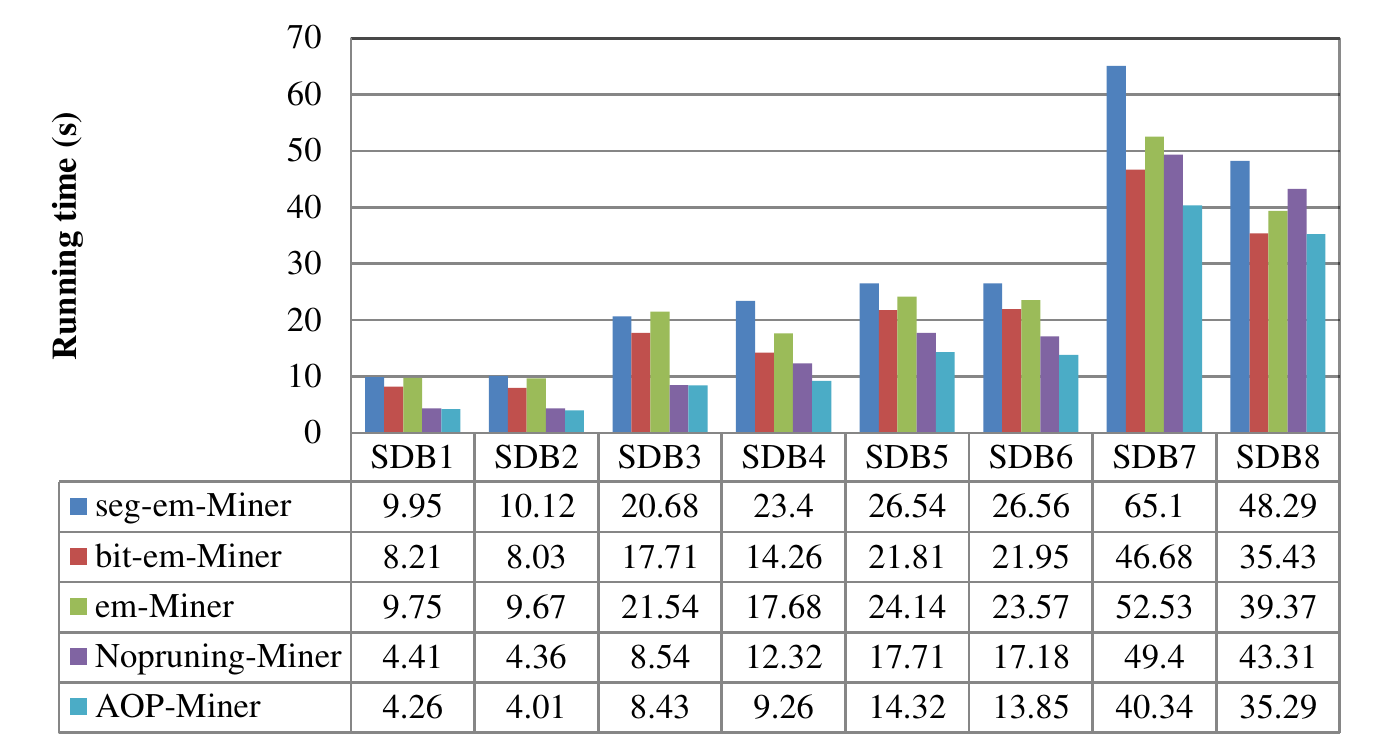}
	\caption{Comparison of number of candidate patterns for SDB1–SDB8}
	\label{fig.4}
\end{figure}

\begin{figure}   
	\setlength{\abovecaptionskip}{0cm} 
	\centering
	\includegraphics[width=1.0\linewidth]{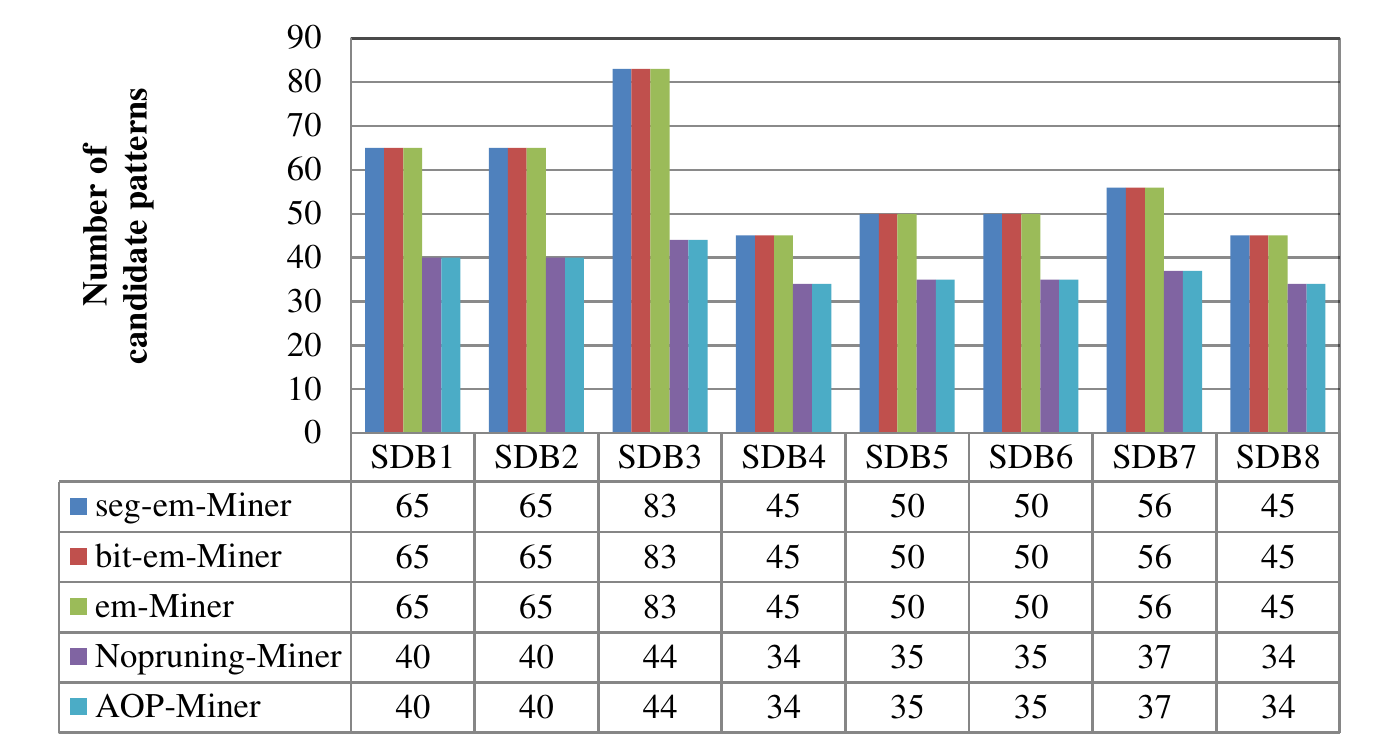}
	\caption{Comparison of running time for SDB1–SDB8 }
	\label{fig.5}
\end{figure}

The results give rise to the following observations.

(1) AOP-Miner runs faster than em-Miner, thus verifying the efficiency of the pattern fusion strategy. The two algorithms were able to find the same number of patterns on all datasets. Moreover, from Fig. \ref{fig.5}, we see that em-Miner took 9.75s for SDB1, while AOP-Miner took 4.26s. Similar results can be found for all other datasets. The reason for this is as follows. Although AOP-Miner and em-Miner adopt the same support calculation method, they use different methods to generate the candidate patterns. AOP-Miner runs faster than em-Miner, meaning that the pattern fusion strategy outperforms the enumeration strategy, since the enumeration strategy generates more candidate patterns. For example, from Fig. \ref{fig.4}, we see that em-Miner generates 65 candidate patterns, while AOP-Miner generates 40. The more the candidate patterns, the slower the algorithm runs. Hence, AOP-Miner outperforms em-Miner.


(2) AOP-Miner runs faster than Nopruning-Miner, thus verifying the effectiveness of the pruning strategy. From Fig. \ref{fig.4}, we know that the two algorithms generate 44 candidate patterns on SDB3. However, the running time for Nopruning-Miner is 8.54s, while that of AOP-Miner is 8.43s. Similar effects can be found for the other datasets. We know that Nopruning-Miner is the same as AOP-Miner except for the pruning strategy. As Example \ref{example10} shows, the pruning strategy can avoid redundant support calculations, and hence, AOP-Miner outperforms Nopruning-Miner.

(3) AOP-Miner runs faster than seg-em-Miner and bit-em-Miner, which further validates the efficiency of both the pattern fusion strategy and the screening and pruning strategies for calculating the supports. From Fig. \ref{fig.4}, we know that seg-em-Miner and bit-em-Miner generate 45 candidate patterns on SDB4, while AOP-Miner generates 34 ones. Moreover, the running time of seg-em-Miner and bit-em-Miner are 23.4s and 14.26s, respectively, while that of AOP-Miner is 9.26s. Similar results can be found for all other datasets. The reason for this is that AOP-Miner employs more effective candidate pattern generation and support calculation methods, and therefore outperforms seg-em-Miner and bit-em-Miner.

In conclusion, AOP-Miner has better performance than competitive algorithms. 

\section{Conclusion}\label{section6}

To mine order-preserving patterns in a flexible way, we have proposed the use of AOP mining and have developed the AOP-Miner algorithm. This algorithm can not only mine patterns with identical trends but also patterns with similar trends. Users can set the level of similarity using the parameters $\delta$ and $\gamma$ to measure the local and global similarities, respectively. For candidate pattern generation, AOP-Miner uses a pattern fusion strategy that significantly reduces the number of candidate patterns. To calculate the pattern support, AOP-Miner employs a screening strategy to find the candidate occurrences of a superpattern based on the occurrences of subpatterns, and adopts a pruning strategy to further reduce the number of candidate patterns. These screening and pruning strategies can greatly lower the number of pattern matches, and significantly improve the mining performance. Experimental results on real-time series datasets indicate that AOP-Miner is an efficient AOP mining algorithm. 

Inspired by OPP mining \cite{ref26}  and  	($\delta$,$\gamma$)-approximate  order-preserving pattern matching \cite{ref44}, this paper investigates  approximate order-preserving pattern mining based on ($\delta$,$\gamma$)-distances. However, there are many similarity measurement methods, such as DTW \cite {keoghdtw}. DTW method can effectively handles shrink, scaling, and noise injection, while  ($\delta$,$\gamma$)-distances cannot. Hence, the research of approximate order-preserving pattern mining based on DTW distance is worth investigating.


\begin{thebibliography}{99}
	\bibitem{ref1} P. Fournier-Viger, W. Gan, Y. Wu, M. Nouioua, W. Song, T. Truong, H. Duong, "Pattern mining: Current challenges and opportunities", in PMDB, 2022, 34-49. 
	
	\bibitem{wutmis2}X. Wu, X. Zhu, M. Wu, "The evolution of search: Three computing paradigms", ACM Transactions on Management Information Systems (TMIS),  vol. 13, no.2, pp. 20, 2022.
	
	\bibitem{ref3}G. Atluri, A. Karpatne, V. Kumar, "Spatio-temporal data mining: A survey of problems and methods", ACM Computing Surveys (CSUR), vol. 51, no.4, pp. 1-41, 2018.
	
	\bibitem {tmisweak} Y. Wu, X. Wang, Y. Li, L. Guo, Z. Li, J. Zhang, X. Wu, "OWSP-Miner: Self-adaptive one-off weak-gap strong pattern mining", ACM Transactions on Management Information Systems, vol. 13, no. 3, pp. 25, 2022.
	
	%
	
	
	\bibitem{ref6} Y. Li, S. Zhang, L. Guo, J. Liu, Y. Wu, X. Wu, "NetNMSP: Nonoverlapping maximal sequential pattern mining", Applied Intelligence, vol. 52, no.9, pp, 9861-9884, 2022.
	
	\bibitem{kbs2020}Y. Wu, C. Zhu, Y. Li, L. Guo, X. Wu, "NetNCSP: Nonoverlapping closed sequential pattern mining", Knowledge-Based Systems, vol. 196, 105812, 2020.
	
	\bibitem{leiduan} T. Wang, L. Duan, G. Dong, Z. Bao, "Efficient mining of outlying sequence patterns for analyzing outlierness of sequence data", ACM Transactions on Knowledge Discovery from Data, vol. 14, no. 5, pp. 62, 2020.
	
	\bibitem{dse}D. Samariya, J. Ma, "A new dimensionality-unbiased score for efficient and effective outlying aspect mining", Data Science and Engineering, vol. 7, no. 2, pp. 120-135, 2022.
	
	
	\bibitem{lizhenwang} L. Wang, X. Bao, L. Zhou, "Redundancy reduction for prevalent co-location patterns", IEEE Transactions on Knowledge and Data Engineering, vol. 30, no. 1, pp. 142-155, 2018.
	
	\bibitem{lizhenwang2} L. Wang, Y. Fang, L. Zhou, "Preference-based spatial co-location pattern mining", Series Title: Big Data Management. Springer Singapore, 2022, https://doi.org/10.1007/978-981-16-7566-9
	
	
	\bibitem{ref7} S. Ghosh, J. Li, L. Cao, K. Ramamohanarao, "Septic shock prediction for ICU patients via coupled HMM walking on sequential contrast patterns", Journal of Biomedical Informatics, vol. 66, pp.  19-31, 2017.
	
	\bibitem{ref13}
	P. Senin, J. Lin, X. Wang, T. Oates, S. Gandhi, A. P. Boedihardjo, C. Chen,  S. Frankenstein, "Grammarviz 3.0: Interactive discovery of variable-length time series patterns", ACM Transactions on Knowledge Discovery from Data (TKDD), vol. 12, no. 1, pp. 1-28, 2018.
	
	
	\bibitem{approximate}
	C. C. M. Yeh, Y. Zheng, J. Wang, H. Chen, Z. Zhuang, W. Zhang,  E Keogh,  "Error-bounded approximate time series joins using compact dictionary representations of time series", Proceedings of the 2022 SIAM International Conference on Data Mining (SDM), 2022, pp. 181-189.
	
	
	
	\bibitem{ref15} J. Zuo, K. Zeitouni, Y. Taher, "SMATE: Semi-supervised spatio-temporal representation learning on multivariate time series", in ICDM, 2021, pp. 1565-1570.
	
	\bibitem{ref18} J. Deng, X. Chen, R. Jiang, X. Song, I. W. Tsang. "ST-Norm: Spatial and temporal normalization for multi-variate time series forecasting", Proceedings of ACM SIGKDD, 2021, pp. 269-278.
	
	\bibitem{Jlin}
	Y. Gao, J. Lin. "Discovering subdimensional motifs of different lengths in large-scale multivariate time series". ICDM, 2019, pp. 220-229.
	
	\bibitem{zlibigdata} 
	Z. Li, J. He, H. Liu, X. Du,  "Combining global and sequential patterns for multivariate time series forecasting", In 2020 IEEE International Conference on Big Data (Big Data), 2020, pp. 180-187.
	
	
	
	
	\bibitem{ref20} I. R. Parray, S. S. Khurana, M. Kumar, A. A. Altalbe, "Time series data analysis of stock price movement using machine learning techniques", Soft Computing, vol. 24, no. 21, pp. 16509-16517, 2020.            
	
	\bibitem{ref21} C. Dai, J. Wu, D. Pi, S. I. Becker, L. Cui, Q. Zhang, B. Johnson. "Brain EEG time-series clustering using maximum-weight clique". IEEE Transactions on Cybernetics, vol.  52, no. 1, pp.  357-371, 2022.
	
	\bibitem{icde}
	N. Alghamdi, L. Zhang, H. Zhang, E. A. Rundensteiner, M. Y. Eltabakh, "ChainLink: Indexing big time series data for long subsequence matching",  IEEE 36th International Conference on Data Engineering (ICDE), 2020, pp. 529-540.
	
	\bibitem{Agrawaltkde}
	S. Agrawal, M. Steinbach, D. Boley, S. Chatterjee, G. Atluri, A. T. Dang, S. Liess,  V. Kumar, "Mining novel multivariate relationships in time series data using correlation networks", IEEE Transactions on Knowledge and Data Engineering, vol. 32, no. 9, pp. 1798-1811, 2019.
	
	
	
	
	\bibitem{youxiins} Y. Wu, Z. Yuan, Y. Li, L. Guo, P. Fournier-Viger, X. Wu, "NWP-Miner: Nonoverlapping weak-gap sequential pattern mining", Information Sciences, vol. 588, pp. 124-141, 2022.
	
	\bibitem{minfan2020} F. Min, Z. Zhang, W. Zhai, R. Shen, "Frequent pattern discovery with tri-partition alphabets". Information Sciences, vol. 507, pp. 715-732,  2020.
	
	
	\bibitem{youxitkdd}Y. Wu, L. Luo, Y. Li, L. Guo, P. Fournier-Viger, X. Zhu, X. Wu, "NTP-Miner: Nonoverlapping three-way sequential pattern mining", ACM Transactions on Knowledge Discovery from Data, vol. 16, no.3, pp. 51,  2022.
	
	\bibitem {ref24}  Y. Lin, I. Koprinska, M. Rana, "SSDNet: State space decomposition neural network for time series forecasting", in ICDM, 2021, pp. 370-378.
	
	\bibitem {gnnAaai}  A. Deng, B. Hooi, "Graph neural network-based anomaly detection in multivariate time series", Proceedings of the AAAI Conference on Artificial Intelligence, vol. 35, no. 5, pp. 4027-4035, 2021.
	
	\bibitem {ref25} H. Zhou, S. Zhang, J. Peng, S. Zhang, J. Li, H. Xiong, W. Zhang, "Informer: Beyond efficient transformer for long sequence time-series forecasting", Proceedings of AAAI, 2021.
	
	
	
	\bibitem{ref26} Y. Wu, Q. Hu, Y. Li, L. Guo, X. Zhu, X. Wu, "OPP-Miner: Order-preserving sequential pattern mining for time series", IEEE Transactions on Cybernetics, 2022, DOI: 10.1109/TCYB.2022.3169327
	
	\bibitem {oprminer}
	Y. Wu, X. Zhao, Y. Li, L. Guo, X. Zhu, P. Fournier-Viger, X. Wu,  "OPR-Miner: Order-preserving rule mining for time series", IEEE Transactions on Knowledge and Data Engineering,  10.1109/TKDE.2022.3224963
	
	\bibitem{ref44} J. Mendivelso, R. Niquefa, Y. Pinzón, G. Hernández, "New algorithms for $\delta$$\gamma$-order preserving matching", Ingeniería, vol. 23, no.2, pp. 190-202, 2018.
	
	\bibitem{ref28} Y. Li, L. Yu, J. Liu, L. Guo, Y. Wu, X. Wu, "NetDPO: (delta, gamma)-approximate pattern matching with gap constraints under one-off condition", Applied Intelligence, 2022, DOI: 10.1007/s10489-021-03000-2
	
	
	
	
	\bibitem{ref29} Y. Wu, Y. Tong, X. Zhu, X. Wu, "NOSEP: Nonoverlapping sequence pattern mining with gap constraints", IEEE Transactions on Cybernetics, vol. 48, no.10, pp. 2809-2822,  2018.
	
	\bibitem{apin2014}Y. Wu, L. Wang, J. Ren, W. Ding, X. Wu, "Mining sequential patterns with periodic wildcard gaps", Applied Intelligence, vol. 41, no. 1, pp. 99-116, 2014. 
	
	\bibitem{apin2022}Y. Wang, Y. Wu, Y. Li, F. Yao, P. Fournier-Viger, X. Wu, "Self-adaptive nonoverlapping sequential pattern mining", Applied Intelligence, vol. 52, no. 6, pp. 6646-6661, 2022.
	
	
	
	\bibitem{ref32}	X. Dong, Y. Gong, L. Cao, "e-RNSP: An efficient method for mining repetition negative sequential patterns", IEEE Transactions on Cybernetics, vol. 50, no. 5, pp. 2084-2096, 2018.  
	
	\bibitem{tkdd2022} Y. Wu, M. Chen, Y. Li, J. Liu, Z. Li, J. Li, X. Wu, "ONP-Miner: One-off negative sequential pattern mining", ACM Transactions on Knowledge Discovery in Data. doi: 10.1145/3549940. 2022.
	
	\bibitem{icdmwpart}
	Y. Chen, P. Fournier-Viger, F. Nouioua, Y. Wu, "Sequence prediction using partially-ordered episode rules", ICDM (Workshops), 2021, pp. 574-580.
	
	\bibitem{ref33}M. M. Hossain, Y. Wu, P. Fournier-Viger, Z. Li, L. Guo, Y. Li, "HSNP-Miner: High utility self-adaptive nonoverlapping pattern mining", ICBK, 2021, pp. 70-77.
	
	\bibitem{ref4} W. Song, L. Liu,  C. Huang, Generalized maximal utility for mining high average-utility itemsets, Knowledge and Information Systems 63 (2021) 2947-2967. 
	
	\bibitem{wuutil}Y. Wu, M. Geng, Y. Li, L. Guo, Z. Li, P. Fournier-Viger, X. Zhu, X. Wu, "HANP-Miner: High average utility nonoverlapping sequential pattern mining", Knowledge-Based Systems, vol.  229, pp. 107361, 2021.
	
	\bibitem{utilityeswa}
	Y. Wu, R. Lei, Y. Li, L. Guo, X. Wu, "HAOP-Miner: Self-adaptive high-average utility one-off sequential pattern mining", Expert Systems with Applications 184 (2021) 115449.
	
	
	
	
	
	
	
	
	
	\bibitem{Mercericdm}
	R. Mercer, S. Alaee, A. Abdoli, N. S. Senobari, S. Singh, A. Murillo, E.  Keogh, "Matrix Profile XXIII: Contrast Profile: A novel time series primitive that allows real world classification", 2021 IEEE International Conference on Data Mining (ICDM). IEEE, 2021, pp. 1240-1245.
	
	
	\bibitem{ref39} Y. Wu, Y. Wang, Y. Li, X. Zhu, X. Wu, "Top-k self-adaptive contrast sequential pattern mining", IEEE Transactions on Cybernetics, 2021,  DOI: 10.1109/TCYB.2021.3082114.
	
	
	
	
	
	
	
	
	\bibitem{truongins} T. C. Truong, H. V. Duong, B. Le, P. Fournier-Viger. "EHAUSM: An efficient algorithm for high average utility sequence mining", Information Sciences, vol. 515, pp. 302-323,  2020. 
	
	\bibitem{ref40}
	E. Keogh, J. Lin, A. Fu. "Hot SAX: Efficiently finding the most unusual time series subsequence", Fifth IEEE International Conference on Data Mining (ICDM'05).  2005, pp. 8.
	
	
	\bibitem{keoghdtw} E. Keogh, C. A. Ratanamahatana, "Exact indexing of dynamic time warping", Knowledge and Information Systems, vol. 7, no. 3, pp. 358-386, 2005.
	
	\bibitem{kdd1994} D. J. Berndt, J. Clifford, "Using dynamic time warping to find patterns in time series", KDD workshop, vol. 10, no. 16, pp. 359-370, 1994.
	
	
	\bibitem{ref41} J. Kim, P. Eades, R. Fleischer, S. H. Hong, C. S. Iliopoulos, k. Park, T. Tokuyama, "Order-preserving matching", Theoretical Computer Science, vol. 525, pp. 68-79, 2014.
	
	
	
	\bibitem{ref43} P. Gawrychowski, P. Uznański. "Order-preserving pattern matching with k mismatches", Theoretical Computer Science, vol. 638, pp. 136-144,  2016.
	
	\bibitem{jinquandg}Y. Wu, J. Fan, Y. Li, L. Guo, X. Wu, "NetDAP: (delta, gamma)-Approximate pattern matching with length constraints", Applied Intelligence, vol. 50, no. 11, pp. 4094-4116, 2020.
	
	
	
	
	
\end{thebibliography}
\end{document}